\documentclass[11pt]{article}

\usepackage[a4paper, margin=1in]{geometry}
\usepackage[english]{babel}
\usepackage[utf8]{inputenc}

\usepackage{lmodern}
\usepackage{graphicx}
\usepackage[colorinlistoftodos]{todonotes}

\usepackage[colorlinks=true]{hyperref}
\usepackage{hyperref}

\usepackage{amsthm}
\usepackage{amsmath}
\usepackage{amssymb}
\usepackage{amsfonts}
\usepackage{mathrsfs}
\usepackage{mathtools}
\usepackage{verbatim}
\usepackage{footnote}
\usepackage{lineno}
\usepackage{caption}
\usepackage{tikzsymbols}
\usepackage[capitalize, nameinlink]{cleveref}

\usepackage[symbol]{footmisc}

\usepackage{icomma}                 
\usepackage{enumitem}              
\usepackage{array}                  
\usepackage{multirow}
\usepackage{setspace}

\usepackage{algorithm}              
\usepackage{algorithmicx}           
\usepackage[noend]{algpseudocode}   
\usepackage{listings}

\floatname{algorithm}{Algorithm}

\usepackage{physics}
\usepackage{euscript}               
\usepackage{xspace}

\usepackage{tikz}
\usepackage{boxedminipage2e}
\usepackage[section]{placeins}

\usepackage{color}

\newtheorem{theorem}{Theorem}
\newtheorem{lemma}{Lemma}
\newtheorem{corollary}{Corollary}
\newtheorem{definition}{Definition}

\newtheorem*{theor}{Theorem}

\newcommand{\N}{\mathbb{N}}

\newcommand{\R}{\mathbb{R}}

\newcommand{\Cc}{\mathcal{C}}

\newcommand{\wF}{\tilde{F}}
\newcommand{\wC}{\tilde{C}}

\newcommand{\hx}{\hat{x}}
\newcommand{\hy}{\hat{y}}

\newcommand{\dx}{\dot{x}}
\newcommand{\dy}{\dot{y}}

\newcommand{\eps}{\varepsilon}

\newcommand{\argmin}{\mathop{\mathrm{argmin}}}

\newcommand{\mlklp}{\operatorname{MLkSC-LP}}
\newcommand{\mslp}{\operatorname{MSSC-LP}}

\newcommand{\MLkSC}{\mathrm{MLkSC}\xspace}
\newcommand{\MSSC}{\mathrm{MSSC}\xspace}

\newcommand{\optmlk}{\mathrm{OPT}_\mathrm{MLk}}
\newcommand{\optms}{\mathrm{OPT}_\mathrm{MS}}
\newcommand{\SC}{\mathrm{SC}\xspace}
\newcommand{\sclp}{\operatorname{SC-LP}}

\newcounter{alphasect}
\def\alphainsection{0}

\let\oldsection=\section
\def\section{%
	\ifnum\alphainsection=1%
	\addtocounter{alphasect}{1}
	\fi%
	\oldsection}%

\renewcommand\thesection{%
	\ifnum\alphainsection=1%
	\Alph{alphasect}%
	\else%
	\arabic{section}%
	\fi%
}%

\date{}
\title{\Large \bf Approximating Star Cover Problems}
\author{
	{\rm Buddhima Gamlath}\\
	buddhima.gamlath@epfl.ch \\
	EPFL, Lausanne, Switzerland
	\and
	{\rm Vadim Grinberg}\\
	vgm@ttic.edu\\
	TTIC\footnote{Part of the work was done while the author was a Summer@EPFL intern in the School of Computer and Communication Sciences, Ecole polytechnique federale de Lausanne, Lausanne, and a full-time undergraduate student in the Faculty of Computer Science, Higher School of Economics, Moscow.}, Chicago, USA
}

\begin{document}

\maketitle              
\begin{abstract}

Given a metric space $(F \cup C, d)$, we consider star covers of $C$ with balanced loads. A star is a pair $(f, C_f)$ where $f \in F$ and $C_f \subseteq C$, and the load of a star is $\sum_{c \in C_f} d(f, c)$. In minimum load $k$-star cover problem $(\MLkSC)$, one tries to cover the set of clients $C$ using $k$ stars that minimize the maximum load of a star, and in minimum size star cover $(\MSSC)$ one aims to find the minimum number of stars of load at most $T$ needed to cover $C$, where $T$ is a given parameter.

We obtain new bicriteria approximations for the two problems using novel rounding algorithms for their standard LP relaxations. For $\MLkSC$, we find a star cover with $(1+\eps)k$ stars and $O(1/\eps^2)\optmlk$ load where $\optmlk$ is the optimum load. For $\MSSC$, we find a star cover with $O(1/\eps^2) \optms$ stars of load at most $(2 + \eps) T$ where $\optms$ is the optimal number of stars for the problem.
Previously, non-trivial bicriteria approximations were known only when $F = C$.

\end{abstract}

\section{Introduction}

Facility location (FL) is a family of problems in computer science where the general goal is to assign a set of clients to a set of facilities under various constraints and optimization criteria. 
FL family encompasses many natural clustering problems like $k$-median and $k$-means, most of which are well studied. 
In this work, we study two relatively less studied FL problems which we call minimum load $k$-star cover ($\MLkSC$) and minimum size star cover $(\MSSC$). The goal of $\MLkSC$ is to assign clients to at most $k$ facilities, minimizing the maximum assignment cost of a facility, while that of $\MSSC$ is to find a client-facility assignment with the minimum number of facilities such that the total assignment cost of each facility is upper bounded by a given threshold $T$.

We begin by formally defining the two problems. Let $C$ be a finite set of clients and $F$ be a finite set of facilities. Let $(F \cup C, d)$ be a finite metric space where $d : (F \cup C) \times (F \cup C) \to \R_0^+$ is a distance metric. By a \emph{star} in $(F, C)$, we mean any tuple $(f, C_f)$, where $f \in F$ and $C_f \subseteq C$.
We say two stars $(f, C_f)$ and $(g, C_g)$ are \emph{disjoint} if $f \neq g$ and $C_f \cap C_g = \varnothing$. A \emph{star cover} of $(F, C)$ is a finite collection $S = \{(f_1, C_{f_1}), \allowbreak \dots, (f_{|S|}, C_{i_{|S|}})\}$ of disjoint stars such that $C = C_{f_1} \cup \dots \cup C_{f_{|S|}}$. The size of a star cover $S$ is the number of stars $|S|$ in the cover. Given a star cover $S$, a star $(f, C_f) \in S$, and a client $c \in C_f$, we say that client $c$ is \emph{assigned} to facility $f$ under $S$ and the facility $f$ is \emph{serving} client $c$ under $S$.
For a star $(f, C_f)$, the \emph{load} of facility $f$ is the sum of pair-wise distances $\sum_{c \in C_f}d(f, c)$ between itself and its clients.
The \emph{load} $L(S)$ of a star cover $S$ is the load of its maximum load star. I.e., $L(S) := \max_{(f, C_f) \in S}\sum_{c \in C_f}d(f, c)$.
For notational convenience, we denote the collection of all star covers of $(F, C)$ by $\mathcal{S}$.
Using the introduced notation, we now define $\MLkSC$ and $\MSSC$.

\begin{definition}[Minimum Load $k$-Star Cover]
    Given a finite metric space $(F \cup C, d)$ and number $k \in \N$, the task of minimum load $k$-star cover problem is to find a star cover of size at most $k$ that minimizes the load; I.e., find
    $S^\ast := \argmin_{S \in \mathcal{S} : |S| \leq k} L(S).$
    We denote the optimal load $L(S^\ast)$ by $\optmlk$.
\end{definition}
\begin{definition}[Minimum Size Star Cover]
    Given a finite metric space $(F \cup C, d)$ and a number $T \in \R_+$, the task of minimum size star cover problem is to find a star cover of load at most $T$ that minimizes the size; I.e., find a star cover $S^\star := \argmin_{S \in \mathcal{S} : L(S) \leq T} |S|$.
    We denote the optimal size $|S^\star|$ by $\optms$.
\end{definition}

Even et al.~\cite{EGK03} showed that both $\MLkSC$ and $\MSSC$ are NP-hard for general metrics even when $F = C$. Both Even et al.~\cite{EGK03} and Arkin et al.~\cite{AHL06} studied the problem in $F = C$ setting and gave constant factor bicriteria approximation algorithms for $\MLkSC$. The latter work also gave a constant factor approximation algorithm for $\MSSC$ in the same setting.

Arkin et al.~\cite{AHL06} use k-median clustering and then split the individual clusters that are too large into several smaller clusters to obtain their approximation guarantees. 
However, the splitting of clusters rely on that the clients and facilities are indistinguishable, which allows one to conveniently choose a new facility for each new partition created in the splitting process. 
Meanwhile, the technique of Even et al.~\cite{EGK03} is to formulate the problem as an integer program, round its LP relaxation using minimum make-span rounding techniques, and use a clustering approach that also relies on $F$ being equal to $C$ to obtain the final bicriteria approximation guarantees. 
Both the techniques do not generalize to the case where $F \neq C$ unless it is allowed to open the same facility multiple times. 

Recently, Ahmadian et al.~\cite{ABB18} showed that $\MLkSC$ is NP-hard even if we restrict the metric space to be a line metric.
They further gave a PTAS for $\MLkSC$ in line metrics and a quasi-PTAS for the same in tree metrics. However, their techniques are specific to line and tree metrics, and it is not known whether they can be extended to general metrics.

The main goal of this work is to extend the approach of Even et al.~\cite{EGK03} to $F \neq C$ setting where any given facility can be opened at most once. 
To do so, we introduce a novel clustering technique and an accompanied new algorithm to modify the LP solution before applying the minimum makespan rounding at the end.  This yields the following theorem:
\begin{theorem}\label{t1}
    There exists a polynomial time algorithm that, given an instance $(F \cup C, d)$ of $\MLkSC$ problem and any $\eps \in (0,1)$, finds a star cover of $(F, C)$ of size at most $(1 + \eps)k$ and load at most $O(\optmlk /\eps^2)$.
\end{theorem}

As a complementary result, we also show that the standard LP relaxation has some inherent limitations. That is, we construct a family of $\MLkSC$ instances where the load of any integral $(1 + \eps)k$-star cover is at least $\Omega(1/\eps)$ times the optimal value of the standard LP.

With slight modifications to our clustering and rounding techniques, we further obtain the following theorem on $\MSSC$:
\begin{theorem}\label{t2}
     There exists a polynomial time algorithm that, given an instance of $\MSSC$ problem with load parameter $T$ and any $\eps \in (0, 1)$, finds a star cover of load at most $(2 + \eps)T$ and size at most $O(\optms/\eps^2)$.
\end{theorem}

As with $\MLkSC$, we show that the standard LP-relaxation for $\MSSC$ also suffers from inherent limitations; I.e., for any $\eps > 0$, we give an instance of $\MSSC$ for which there is a \emph{fractional} star cover of load at most $T$ but any integral star cover of that instance has load at least $(2 - \eps)T$ even with all facilities opened.

We end the introduction with a brief section on other related work. In \cref{sec:tech}, we introduce the LP relaxations of the two problems and provide a more elaborate description of our techniques. Later in \cref{sec:algo1} and \cref{sec:algo2} we describe the proofs of \cref{t1} and \cref{t2} in detail. We present the explicit constructions of families of $\MLkSC$ and $\MSSC$ that show inherent limitations of the respective standard LP relaxations in \cref{sec:appb}.

\subsubsection*{Other Related Work}

To the best of our knowledge, Even et al.~\cite{EGK03} and Arkin et al.~\cite{AHL06} were among the first to explicitly address close relatives of $\MLkSC$ and $\MSSC$ problems. 
Both of their works considered the problem where one has to cover nodes (or edges) of a graph using a collection of \emph{objects} (I.e.,trees or stars). 
Evans et al. considered the problem of minimizing the maximum cost of an object when the number of objects is fixed, for which they gave a $4$-approximation algorithm.  
Arkin et al. also studied the same problem and additionally considered paths and walks as covering objects. 
They further discussed the $\MSSC$ version of the problems where the goal is to minimize the number of covering objects such that the cost of each object is at most a given threshold. 
For min-max tree cover with $k$ trees, Khani and Salavatipour~\cite{KS11} later improved the approximation guarantee to a factor of three.

In general, many well-known facility location problems have constant factor approximation guarantees. 
For example, for uncapacitated facility location, the known best algorithm (Li et al.~\cite{Li13}) gives an approximation ratio of 1.488. For $k$-median in general metric spaces, the current best is $2.675$ due to Byrka et al.~\cite{BTS17}, and for $k$-means in general metric spaces, it is $(9 + \varepsilon)$  due to Ahmadian et al.~\cite{ANSW18}. Remarkably, all these results follows from LP based approaches.
A common theme of all these problems is that their objectives are to minimize a summation of costs. I.e., we minimize the sum of distances from clients to their respective closest opened facilities, where in uncapacitated facility location problem, we additionally have the sum of opening costs of the opened facilities. This \emph{min-sum} style objective is in contrast with the min-max style objective of minimum star cover problem which makes it immune to algorithmic approaches that are applicable to other common facility location counterparts.
 
As discussed, minimum star cover problems are closely related to minimum makespan scheduling and the generalized assignment problem. 
Two most influential literature in this regard include  Lenstra et al.~\cite{LST90} and Shmoys et al.~\cite{ST93}.

\section{Our Results and Techniques}
\label{sec:tech}

We start with the LP relaxations of the standard integer program formulations for $\MLkSC$ and $\MSSC$. To make the presentation easier, we first define a polytope $\sclp(T, k)$ such that the \emph{integral} points of $\sclp(T, k)$ are feasible star covers of load at most $T$ and size at most $k$.

For $i \in F$, let variable $y_i \in \{0, 1\}$ denote whether $i$'th facility is \emph{opened} (I.e., $y_i=1$ if and only if there is a star $(i, C_i)$ in the target star cover), and for $(i, j) \in F \times C$, let variable $x_{ij} \in \{0, 1\}$ denote whether $j$'th client is \emph{assigned} to facility $i$ (I.e., $x_{ij} = 1$ if and only if $j \in C_i$ where $(i, C_i)$ is a star in the target star cover). Then the following set of constraints define $\sclp(T, k)$:
\[\begin{minipage}[c]{0.7\textwidth}
\begin{align}
&  &&  \sum_{j \in C}d(i, j) \cdot x_{ij} \leq T \cdot y_i  & \forall i \in F, \label{eq:scons1}\\
&                    &&  \sum_{i \in F}y_i \leq k, \label{eq:scons2}\\
&                    &&  \sum_{i \in F}x_{ij} = 1 & \forall j \in C, \label{eq:scons3}\\
&                    &&  x_{ij} \leq y_i & \forall i \in F, \forall j \in C, \label{eq:scons4}\\
&                    &&  y_i \in [0, 1]  & \forall i \in F, \label{eq:scons5}\\
&                    &&  x_{ij} \in [0, 1]  & \forall i \in F, \forall j \in C \label{eq:scons6}\\
&                    && x_{ij} = 0&\forall i \in F, \forall j \in C: d(i, j) > T. \label{eq:scons7}
\end{align}
\end{minipage}
\tag{$\sclp(T, k)$}\]

Here, Constraint~\eqref{eq:scons1} ensures that the load of an opened facility $i \in F$ is at most $T$, 
while Constraint~\eqref{eq:scons2} limits the maximum number of opened facilities to $k$.
Constraint~\eqref{eq:scons3} and Constraint~\eqref{eq:scons4} ensure that each client is fully assigned
and they are only assigned to opened facilities. Finally Constraint~\eqref{eq:scons5} and 
Constraint~\eqref{eq:scons6} ensures that the only integral values of $x_{ij}$'s and $y_i$'s are $0$ or $1$, while Constraint~\eqref{eq:scons7} essentially removes any $(i, j)$ pair from consideration if the distance between them is larger than $T$.

Note that we can now define the LP for $\MLkSC$ as
\[ \text{Minimize } T \text{ such that } \sclp(T, k) \text{ is feasible,} \tag{$\mlklp$} \]
where one can find the minimum such $T$ using the standard binary search technique.
Similarly, the LP for $\MSSC$ can be stated as
\[ \text{Minimize } k \text{ such that } \sclp(T, k) \text{ is feasible.} \tag{$\mslp$} \]
Recall that $k$ is part of the  $\MLkSC$ problem input and $T$ is a part of the $\MSSC$ problem input.

For an arbitrary (not necessarily feasible) solution $(x, y)$ to $\sclp(T, k)$, for $i \in F$ let $L(i, x)$ denote the \emph{fractional load} of facility $i$ with respect to the assignment $x$, I.e., $L(i, x) := \sum_{j \in C}d(i, j)x_{ij}$.
A solution $(x, y)$ to $\sclp(T, k)$ is called \emph{$(\alpha, \beta)$-approximate}, if for every $i \in F$, $L(i, x) \leq \alpha T y_i$, and $\sum_{i \in F}y_i \leq \beta k$.
The proofs of \cref{t1} and \cref{t2} immediately follow from the two theorems on rounding feasible solutions of $\sclp$ presented below:

\begin{theorem}\label{scload}
    There exists a polynomial time rounding algorithm that, given
    a feasible solution $(x^\ast, y^\ast)$ to $\sclp(T, k)$ and any $\eps \in (0, 1)$, outputs an integral $(O(1/\eps^2), 1 + \eps)$-approximate solution to $\sclp(T, k)$.
\end{theorem}

\begin{proof}[Proof of \cref{t1}]
 Let $\eps \in (0, 1)$ be given.
 Using standard binary search approach, we can guess the value $T^*$, such that $\optmlk \leq T^* \leq 2\optmlk$, by solving $\mlklp$ multiple times for different values of $T^*$ and either finding a feasible fractional solution of load at most $T^*$, or determining that no such solution exists.
 Let $(x^*, y^*)$ be the corresponding fractional solution to $\mlklp$.
 Observe that $(x^*, y^*)$ is a feasible solution to $\sclp(T^*, k)$.
 By \cref{scload}, we can round $(x^*, y^*)$ to an integral solution $(\dx, \dy)$, which opens at most $(1 + \eps)k$ facilities and achieves maximum load at most $O(1/\eps^2)T^*$, and it will take polynomial time.
 Therefore, $(\dx, \dy)$ will be an integral solution to $\mlklp$ with opening at most $(1 + \eps)k$ and maximum load at most $O(1/\eps^2)\optmlk$.
\end{proof}
\begin{theorem}\label{scsize}
    There exists a polynomial time rounding algorithm that, given
    a feasible solution $(x^\ast, y^\ast)$ to $\sclp(T, k)$ and any $\eps \in (0, 1)$, outputs an integral $(2 + \eps, O(1/\eps^2))$-approximate solution to $\sclp(T, k)$.
\end{theorem}
The proof of \cref{t2} using \cref{scsize} is just the same as the proof of \cref{t1} using \cref{scload}, omitting the binary search part (as we optimize over $k$ instead of $T$).

Note that $\mlklp$ closely resembles the LP used in minimum make-span rounding by Lenstra et al.~\cite{LST90}. 
In fact, for the case where we do not have a restriction on number of opened facilities, we can 
assume $y_i = 1$ for all $i \in F$, and the LP reduces to the minimum make-span problem, 
yielding a $2$-approximation algorithm.
The main difficulty here is to figure out which facilities to open. 
Once we have an integral opening of facilities, we can still use minimum make-span rounding at a 
loss of only a factor two in the guarantee for minimum load.
Thus, our algorithm for $\MLkSC$ essentially transforms the initial solution for $\mlklp$ via a 
series of steps to a solution with integral openings, I.e., $y_i \in \{0, 1\}$ for all $i \in F$, 
and fractional assignments, without violating Constraint~\ref{eq:scons1} by too much.

When we fully open (I.e., set $y_i = 1$) some facilities in the solution, inevitably, we have to close down (set $y_i = 0$) some other partially opened ones, which requires redistributing their assigned clients to the opened ones. 
This process is called \emph{rerouting} and is a well-known technique in rounding facility-location-like problems.
However, instead of bounding the total load of all facilities, our problem requires bounding each $L(i, x)$ separately, and consequently, many facility-location rounding algorithms which use rerouting fail to produce a good solution.

Let $x^\circ$ be the solution we obtain from $x$ after rerouting facility $i$ to facility $h$.
Using  triangle inequality $d(h, j) \leq d(h, i) + d(i, j)$ for $j \in C$, we can bound $L(h, x^\circ)$, the new load of $h$:
\[L(h, x^\circ) \leq L(h, x) + L(i, x) + d(h, i)\sum_{j \in C}x_{ij}.\] 
If both $L(h, x)$ and $L(i, x)$ were initially $O(T)$, the new load of $h$ will also be $O(T)$ if and only if the sum $d(h, i)\sum_{j \in C}x_{ij} \leq d(h, i)|N(i)|$ is also at most $O(T)$ (here $N(i)$ is the set of all clients partially served by $i$). However, if $d(h, i)|N(i)|$ is large for all other facilities $h$, a good alternative to rerouting is to open $i$ integrally and assign every client in $N(i)$ to $i$. We call such facilities \emph{heavy} facilities. There is still an issue if the integral load $\sum_{j \in N(i)}d(i, j)$ is too large compared to $T$, but we show that we can prevent having too large integral loads in heavy facilities by preceding the rerouting step with additional filtering and preprocessing steps. 
The filtering step blows-up the load constraint by a $(1 + \eps)$ factor while ensuring that no client is fractionally assigned to far away facilities.
The preprocessing step uses techniques similar to those of minimum make-span rounding by Lenstra et al.~\cite{LST90} to ensure that any non-zero fractional assignment $x_{ij}$ to a facility $i$ is at least a constant factor times its opening $y_i$, while slightly relaxing other constraints.

Once we identify the heavy facilities, we cluster the remaining, non-heavy facilities, and choose which ones should be opened based on the clustering.
Then we redistribute the assignments of the remaining facilities to those that were opened.
Using the properties of the preprocessed solution and the clustering, and using the fact that none of the un-opened facilities are heavy, we show that the resulting fractional assignment satisfies the constraints up to an $O(1/\eps^2)$ factor violation of load constraints.
Hence, the algorithmic result of \Cref{scload} follows from the minimum make-span rounding of Lenstra et al.~\cite{LST90}, which gives us an integral assignment with maximum load increased at most by another factor of $2$


The algorithm for $\MSSC$ problem, on a high level, resembles that for $\MLkSC$: 
We first alter the solution of $\mslp$ to have integral $y_i$'s and fractional $x_{ij}$'s, allowing
the total opening $\sum_{i \in F}y_i$ to be at most $O(1/\eps^2)$ factor larger than the value of $\mslp$, and then
use minimum make-span rounding of Lenstra et al.~\cite{LST90} to obtain the final solution.
However, since make-span rounding guarantees only a factor two violation in the load constraint,
we need to make sure that our modified solution with integral openings and fractional assignments
introduces only small error in load constraints. Namely, to ensure that the final solution satisfies 
$(2 + \eps)T$ maximum load, before applying the minimum make-span rounding, all the loads must be
at most $(1 + \eps/2)T$.
We ensure this by re-arranging the steps of the algorithm for $\MLkSC$ and carefully choosing the parameters.

\section{$(O(1/\eps^2), 1 + \eps)$-approximation to $\SC(T, k)$}
\label{sec:algo1}

In this section, we show how to convert a (feasible) fractional solution $(x, y)$ of $\sclp$ in to a $(O(1/\eps^2), 1 + \eps)$-approximate solution with integral $y$ values. This together with minimum make-span rounding scheme by Lenstra et al.~\cite{LST90} proves \cref{scload}.

\subsection{Preprocessing and filtering}

Suppose that for each $(i, j) \in F \times C$ we either have $x_{ij} = 0$ or $x_{ij} \geq \gamma y_i$ for constant $\gamma \in (0, 1)$.
Then, if $L(i, x) = \sum_{j \in C}d(i, j)x_{ij} \leq \nu Ty_i$ for some constant $\nu \geq 1$, we have $\sum_{j \in N(i)}d(i, j) \leq \frac{\nu}{\gamma}T$.
Therefore, if we open $i$ integrally and assign all $N(i)$ to $i$, the resulting load of $i$ will be $O(T)$.
Even though we cannot guarantee the property above for every solution $(x, y)$ to $\sclp(T, k)$, we can modify $(x, y)$ so that all non-zero assignments $x_{ij}$ satisfy $x_{ij} \geq \gamma y_i$ for some constant $\gamma \in (0, 1)$ at the expense of slightly relaxing other constraints of $\sclp$.
This is exactly the statement of the preprocessing theorem.
\begin{theorem}[Preprocessing]\label{preproc1}
    Let $(x, y)$ be  such that, for all $i \in F$, $L(i, x) \leq \mu T y_i$ for some constant $\mu \geq 1$ and all other constraints of $\sclp(T, k)$ on variables $x$ are satisfied.
    There exists a polynomial time algorithm that, given such solution $(x, y)$ and a constant $\gamma \in (0, 1)$, finds a solution $(x', y')$ such that
    \begin{enumerate}
        \item $y' = y$, and if $x_{ij} = 0$, then $x'_{ij} = 0$;
        \item for every $(i, j) \in F\times C$, $y'_i \geq x'_{ij}$, and if $x'_{ij} > 0$, then $x'_{ij} \geq \gamma y_i'$;
        \item for every $j \in C$, $1 \geq \sum_{i \in F}x'_{ij} \geq 1 - \gamma$;
        \item for every $i \in F$, $L(i, x') \leq (\mu + 2 - \gamma)Ty_i'$.
    \end{enumerate}
\end{theorem}
That is to say, we can guarantee the property $\{x_{ij} > 0 \iff x_{ij} \geq \gamma y_i\}$ by loosing at most $\gamma$ portion of each client's demand and slightly increasing each facility's load.
Loosing a factor of $\gamma$ demand is affordable for our purposes, as one can meet the demand constraint by scaling each $x_{ij}$ by a factor of at most $1/(1 - \gamma)$. 
Since $\gamma$ is a constant, this would blow up the load constraint only by an additional constant factor.
The proof of \cref{preproc1} is rather technical and is given in \cref{sec:appa}.

We now present our rounding algorithm step by step.
Let $(x, y)$ be a feasible fractional solution to $\sclp(T, k)$ and let $\eps \in (0, 1)$. Let $(\dx, \dy)$ denote the final rounded solution with integral $\dy$ and fractional $\dx$. 
\begin{definition}
    For $j \in C$, let $D(j) := \sum_{i \in F}d(i, j)x_{ij}$, the average facility distance to client $j$.
\end{definition}
Let $\rho := \frac{1 + \eps}{\eps}$.
By applying the well-known filtering technique of Lin and Vitter \cite{LV92} to $(x, y)$, we construct a new solution $(\hx, \hy)$, such that $\sum_{i \in F}\hy_i \leq (1 + \eps)k$, $L(i, \hx) \leq (1 + \eps)T\hy_i$ for all $i \in F$, and for every $i, j$, $\hx_{ij} \leq \hy_i$ and if $\hx_{ij} > 0$, then $d(i, j) \leq \rho D(j)$.
Applying \cref{preproc1} to $(\hx, \hy)$, we obtain solution $(x', y')$ such that
\begin{enumerate}
        \item $\sum_{i \in F}y_i' \leq (1 + \eps)k$,
        \item for all $(i, j)$, $y'_i \geq x'_{ij}$, and if $x'_{ij} > 0$, then $x'_{ij} \geq \gamma y_i'$ and $d(i, j) \leq \rho D(j)$,
        \item for every $j \in C$, $1 \geq \sum_{i \in F}x'_{ij} \geq 1 - \gamma$, and
        \item for every $i \in F$, $L(i, x') \leq (\mu + 2 - \gamma)Ty_i' = \nu Ty_i'$.
\end{enumerate}
Here $\nu := (\mu + 2 - \gamma)$ is a new load bound.
We choose $\mu := (1 + \eps)$ and $\gamma := \eps/(1 + \eps)$, but will keep the parameters unsubstituted, for convenience.
It is easy to see from the bounds above that for every $j \in C$, $\sum_{i \in F:x'_{ij} > 0}y_i' \geq 1/(1 + \eps)$.

\subsection{Opening heavy facilities}

We now give an algorithm to choose heavy facilities based on $(x', y')$. 
\begin{definition}
    For $F'\subseteq F$, $C'\subseteq C$, let $N'(i) :=  \{j \in C' : x'_{ij} > 0\}$, $N'(j) := \{i \in F' : x'_{ij} > 0\}$.
\end{definition}
The algorithm internally maintains two subsets $F' \subseteq F$ and $C' \subseteq C$.
Notice that $N'$ changes as the algorithm modifies $F'$ and $C'$.
\begin{definition}
    A facility $i \in F'$ is \emph{$\lambda$-heavy} for $\lambda > 0$, if $\sum_{j \in N'(i)}D(j) > \lambda T$.
\end{definition}
\Cref{heavy} opens all $\lambda$-heavy facilities for the given value of $\lambda$.
It starts with $F' = F$ and $C' = C$ and scans $F'$ for $\lambda$-heavy facilities.
It fully opens every $\lambda$-heavy facility $i \in F'$ and assigns all $N'(i)$ integrally to $i$.
Then, it discards $i$ from $F'$ and $N'(i)$ from $C'$, and continues until all facilities are processed.
\begin{algorithm}[H]
\caption{Opening Heavy Facilities}\label{heavy}
	\begin{algorithmic}[1]
		\Require{A solution $(x', y')$, $\lambda > 0$.}
		\Ensure{Partial solution $(\dx, \dy)$, sets $F'$, $C'$, such that $\sum_{j \in N'(i)}D(j) \leq \lambda T$, $\forall i \in F'$.}
		\State Initialize $F' \leftarrow F$, $C' \leftarrow C$
		\For{$i \in F'$}
		\If{$\sum_{j \in N'(i)}D(j) > \lambda T$}
		\State Initialize $C(i) \leftarrow N'(i)$
		\State $F' \leftarrow F' \setminus \{i\}$, $\dy_i = 1$
		\For{$j \in C(i)$}
		    \State $C' \leftarrow C'\setminus \{j\}$, $\dx_{ij} = 1$
		    \For{$h \in F\setminus \{i\}$}
		        \State $\dx_{hj} = 0$
		    \EndFor
		\EndFor
		\EndIf
		\EndFor
		\Return $(\dx, \dy)$, $F'$, $C'$.
	\end{algorithmic}
\end{algorithm}
Since for each $h \in F'$ we may discard some clients from $N'(h)$ after every step, facilities that were $\lambda$-heavy might become non-$\lambda$-heavy under updated $F'$ and $C'$.
\cref{heavylemma} shows that this procedure does not open too many facilities and that the load of opened facilities does not exceed $T$ by too much.
\begin{lemma}\label{heavylemma}
    Let $F', C'$ be the sets returned by \cref{heavy}.
    Then $|F \setminus F'| \leq k/\lambda$, and for each facility $i \in F\setminus F'$, $L(i, \dx) \leq \frac{\nu}{\gamma}T$.
\end{lemma}
\begin{proof}
    The set $F\setminus F'$ is exactly the set of facilities integrally opened during \cref{heavy}.
    For $i \in F\setminus F'$, set $C(i)$ in \cref{heavy} is exactly the set of clients, integrally assigned to $i$ by the algorithm.
	Observe that for every $i, h \in F\setminus F'$, $i \neq h$, the sets $C(i)$ and $C(h)$ are \emph{disjoint}.
	Hence, by feasibility of $(x, y)$,
	\[|F\setminus F'|\cdot \lambda T < \sum_{i \in F\setminus F'}\sum_{j \in C(i)}D(j) \leq  \sum_{j \in C}D(j) =  \sum_{j \in C}\sum_{i \in F}d(i, j)x_{ij} \leq \sum_{i \in F}Ty_i \leq Tk\]
	and $|F\setminus F'| < \frac{Tk}{\lambda T} = \frac{k}{\lambda}$.
	Next, by the properties of solution $(x', y')$:
	\[\nu Ty_i' \geq \sum_{j \in C(i)}d(i, j)x'_{ij} \geq  \gamma y_i'\sum_{j \in C(i)}d(i, j) \implies L(i, \dx) = \sum_{j \in C(i)}d(i, j)  \leq \frac{\nu}{\gamma}T. \] 
\end{proof}
We apply \cref{heavy} with $\lambda := 1/\eps$, and by \cref{heavylemma} this opens at most $\eps k$ additional facilities.
The load of each opened facility is at most $\frac{\nu}{\gamma}T = \frac{(1 + \eps)(\mu + 2 - \gamma)}{\eps}T = O(T/\eps)$.
For the returned sets $F'$ and $C'$, $\sum_{j \in N'(i)}D(j) \leq T/\eps$ for all $i \in F'$.
Moreover, since $j \in C'$ if and only if $j$ was not served by any $\lambda$-heavy facility (which got opened), for all $j \in C'$ we have $\sum_{i \in N'(j)}y_i' = \sum_{i \in F: x'_{ij} > 0}y_i'\geq 1/(1 + \eps)$.
Facilities in $F\setminus F'$ are all integral, and it remains to find the integral opening among facilities in $F'$.

As discussed earlier, if we reroute $i \in F'$ to $h \in F'$, to guarantee a good approximation we have to bound the term $d(h, i)|N'(i)|$.
Observe that $\sum_{j \in N'(i)}D(j)$ is an upper bound for $|N'(i)|\min_{j \in N'(i)}D(j)$.
Therefore, to get a good bound, we need to choose $h$ for $i$ so that $d(h, i)$ is at most some constant times $\min_{j \in N'(i)}D(j)$.
This requires some sophisticated clustering technique and a wise choice of facility $h$ for every such $i$.

\subsection{Clustering}

To create an integral opening over $F'$, we partition $F'$ into disjoint clusters, open some facilities in every cluster and reroute the closed ones into opened ones within the same cluster.
Our goal is to cluster $F'$ so that, if $i$ and $h$ belong to the same cluster and we reroute $h$ to $i$, $d(h, i) \leq O(\min_{j \in N'(i)}D(j))$.
Classic clustering approaches for facility-location-like problems do not work, and to achieve this bound we are required to design a novel approach.

Let $\Cc \subseteq C'$ be the set of cluster centers.
For every $j \in \Cc$, let $F'(j) \subseteq F'$ be the set of facilities belonging to the cluster centered at $j$, for $i \in F'$ let $\Cc(i)$ be the center of the cluster $i$ belongs to (I.e., $i \in F'(j) \iff \Cc(i) = j$).

The clustering procedure works as follows.
First, we form cluster centers $\Cc$ by scanning $j \in C'$ in ascending order of $D(j)$ and adding $j$ to $\Cc$ only if there are no other centers in $\Cc$ within the distance $2\rho D(j)$ from $j$.
Having determined $\Cc$, we add facilities from $F'$ to different clusters.
Most classical clustering approaches would put $i$ into $F'(s)$, if $s$ is closest to $i$ among $\Cc$.
Our approach is different: if $i \in F'$ is serving some $s \in \Cc$, we add $i \in F'(s)$ regardless the distance $d(i, s)$.
Otherwise, we consider $j \in N'(i)$ with \emph{minimum} $D(j)$ ($j$ is not a cluster center), take $s \in \Cc$ that prevented $j$ from becoming a center, and add $i$ to $F'(s)$. \Cref{fig:clustering} visualizes the clustering procedure and \cref{alg:clustering} gives its pseudocode.
\begin{algorithm}[H]
\caption{Clustering}\label{clustering}
	\begin{algorithmic}[1]
		\Require{Solution $(x', y')$, sets $F'$ and $C'$.}
		\Ensure{Centers $\Cc\subseteq C'$ and disjoint clusters $F'(s)$ for every $s \in \Cc$, $\sqcup_{s \in \Cc}F'(s) = F'$.}
		\State Initialize $\Cc \leftarrow \varnothing$, sort $j \in C'$ by the values of $D(j)$ in ascending order
		\For{$j \in C'$}
			\If{$\forall s \in \Cc : d(s, j) >  2\rho D(j)$}
				\State $\Cc \leftarrow \Cc \cup \{j\}$
			\EndIf
		\EndFor
		\State For all $s \in \Cc$, initialize $F'(s) \leftarrow \varnothing$
		\For{$i \in F'$}
			\If{$\exists s \in \Cc : i \in N'(s)$}
			    \State $\Cc(i) \leftarrow s$, $F'(s) \leftarrow F'(s) \cup \{i\}$
			\Else
				\State Let $j := \argmin_{r \in N'(i)}D(r)$
				\State Take $s \in \Cc$, such that $D(s) \leq D(j)$, $d(s, j) \leq 2\rho D(j)$ 
				\State $\Cc(i) \leftarrow s$, $F'(s) \leftarrow F'(s) \cup \{i\}$
			\EndIf 
		\EndFor
		\Return $\Cc$, $F'(s)$ for $s \in \Cc$.
	\end{algorithmic}
	\label{alg:clustering}
\end{algorithm}
\begin{figure}[H]
\begin{center}
	\begin{tikzpicture}
	
	\draw[very thick] (0,0) ellipse (2.5cm and 2cm);
	\node[above] at (0,1.25) {$N'(v)$};
	
	\draw (0,0) node[draw,circle,fill=red,minimum size=10pt,inner sep=0pt] (v) [label=above left:$v$] {};
	
	\draw[very thick, dashed] (3, 0) ellipse (1.5cm and 2.25cm);
	\node[above] at (3,1.5) {$N'(j)$};
	\draw (3,0) node[draw,circle,fill=red,minimum size=10pt,inner sep=0pt] (j) [label=above:$j$] {};
	
	\draw[very thick] (6, 0.75) ellipse (2.5cm and 1.5cm);
	\node[above] at (6, 1.5) {$N'(s)$};
	\draw (6, 0.75) node[draw,circle,fill=red,minimum size=10pt,inner sep=0pt] (s) [label=right:$s$] {};

	\draw(-1.1,-1.3) node[draw,fill=green,minimum size=10pt,inner sep=0pt] (u) [label=right:$u$] {};
	\draw[very thick, ->] (u) -- (v);
	
	\draw (1.9,0.4) node[draw,fill=green,minimum size=10pt,inner sep=0pt] (w) [label=above:$w$] {};
	\draw[very thick, ->] (w) -- (v);
	
	\draw (4,0.6) node[draw,fill=green,minimum size=10pt,inner sep=0pt] (i) [label=above:$i$] {};
	\draw[very thick, ->] (i) -- (s);
	
	\node[below] at (6, -1.75) {$D(v) \leq D(j) \leq D(s)$};
	
	\draw (3.5,-1.4) node[draw,fill=green,minimum size=10pt,inner sep=0pt] (h) [label=above right:$h$] {};
	\draw[very thick, dashed, ->] (h) -- (j);
	\draw[very thick, dashed, ->] (i) -- (j);
	\draw[very thick, ->] (h) -- (v);
	\end{tikzpicture}
	\caption{Here, $v, j, s \in C'$, $v, s \in \Cc$, $j\notin \Cc$, $u, w, h, i \in F'$. The bold arrow shows that a facility belongs to the cluster centered at that client, the dashed arrow shows that a particular client has minimum average distance among all clients served by a facility.}
\label{fig:clustering}
\end{center}
\end{figure}
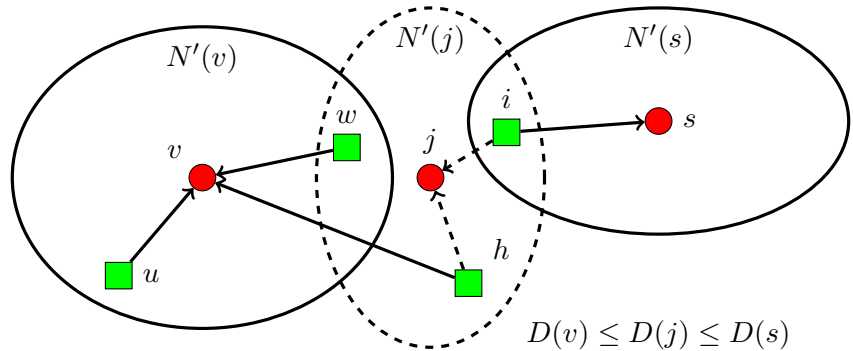
One can easily check that, after \cref{clustering} finishes, for any $s, v \in \Cc$, $s \neq v$, $d(s, v) > 2\rho\max\big(D(s), D(v)\big)$, and as a result $N'(s)$ and $N'(v)$, as well as $F'(s)$ and $F'(v)$ are disjoint.
Also, for every $j \notin \Cc$ there exists $s \in \Cc$ such that $D(s) \leq D(j)$ and $d(s, j) \leq 2\rho D(j)$, simply by construction of the algorithm.
\cref{clustering} allows us to obtain an upper bound on the distance between a facility an its cluster center, represented in terms of minimum average distance of the client served by this facility.
\begin{lemma}\label{distance}
	Let $i \in F'$, let $j = \argmin_{r \in N'(i)}D(r)$.
	Then $d(i, \Cc(i)) \leq 3\rho D(j)$.
\end{lemma}
\begin{proof}
    Let $\Cc(i) = s$.
    There are two cases to distinguish.
    \begin{itemize}
        \item $i \notin N'(s)$ (this case is shown by clients $j, v$ and facility $h$ in \Cref{fig:clustering}).
        By construction of \cref{clustering}, client $s$ is exactly the one that prevented $j$ from becoming a cluster center, therefore $D(s) \leq D(j)$ and $d(j, s) \leq 2 \rho D(j)$.
        Thus, by triangle inequality $d(i, s) \leq d(i, j) + d(j, s) \leq \rho D(j) + 2\rho D(j)= 3\rho D(j)$.
        \item $i \in N'(s)$.
        Then $D(j) = \min_{r \in N'(i)}D(r) \leq D(s)$.
        If $s = j$ or $D(s) = D(j)$, then $d(i, s) \leq \rho D(j)$ automatically.
        Suppose that $s \neq j$, and $D(j) < D(s)$, then $j \notin \Cc$, as $s \in \Cc$ and $i \in N'(j) \cap N'(s)$
        (this case is shown by clients $j, s$ and facility $i$ in \Cref{fig:clustering}).
        Hence, there exists some $s' \in \Cc$ that prevented $j$ from becoming a cluster center, so $D(s') \leq D(j)$ and $d(s', j) \leq 2\rho D(j)$.
        It is easy to see that $D(j)$ must be strictly greater than zero, and since both $s$ and $s'$ are cluster centers, $d(s', s) > 2\rho D(s)$.
        So, by triangle inequality,
        \[2\rho D(s) < d(s', s) \leq  d(s', j) + d(i, j) + d(i, s) \leq  2\rho D(j) + \rho D(j) + \rho D(s),\]
        implying
        \[2\rho D(s) \leq 3\rho D(j) + \rho D(s) \implies D(s) \leq 3D(j).\]
        Since $i \in N'(s)$, it immediately follows that $d(i, s) \leq \rho D(s) \leq 3 \rho D(j)$.
    \end{itemize}
\end{proof}
By applying the triangle inequality once more, we get the desired upper bound on the distances between any two facilities within the same cluster.
\begin{corollary}\label{dist}
    Let $i, h \in F'$, such that $\Cc(i) = \Cc(h)$.
    Let $j = \argmin_{r \in N'(i)}D(r)$ and $v = \argmin_{w \in N'(h)}D(w)$.
    Then $d(i, h) \leq 6\rho \max\big(D(j), D(v)\big)$.
\end{corollary}
Another useful observation is that $\sum_{i \in F'(s)}y_i' \geq 1/(1 + \eps)$ for every cluster center $s \in \Cc$.
It follows from $N'(s) \subseteq F'(s)$ and $\sum_{i \in N'(s)}y_i' \geq 1/(1 + \eps)$.

\subsection{Rerouting}

The last part of our rounding algorithm is opening some facilities in every cluster and rerouting the closed ones.
For $s \in \Cc$ we open $\lfloor (1 + \eps)\sum_{u \in F'(s)}y_u'\rfloor$ facilities in cluster $F'(s)$, prioritizing facilities $i$ with minimum values of $\min_{r \in N'(i)}D(r)$.
Since $\sum_{u \in F'(s)}y_u' \geq 1/(1 + \eps)$, we will open at least one facility in every cluster $F'(s)$ for $s \in \Cc$.
Then, the demand of each closed facility in $F'(s)$ is redistributed at an equal fraction between all the opened ones in $F'(s)$, I.e. we reroute it to all opened facilities in $F'(s)$.
This gives us an integral opening $\dy$ over facilities in $F'$ and a fractional assignment $\dx$ over clients in $C'$.

\cref{reroutlemma} shows that by opening $|K_s| = \lfloor (1 + \eps)\sum_{u \in F'(s)}y_u'\rfloor$ facilities in cluster $F'(s)$ and rerouting all closed facilities in $F'(s)$, we open at most $(1 + 3\eps)k$ facilities in total, and the load of every opened facility in $F'$ exceeds $T$ at most by a constant factor.
\begin{algorithm}[H]
	\caption{Rerouting}\label{rerouting}
	\begin{algorithmic}[1]
	    \Require{Solution $(x', y')$, cluster centers $\Cc$ and clusters $F'(s)$ for $s \in \Cc$.}
	    \Ensure{Solution $(\dx, \dy)$, sets of opened facilities $K_s$ for $s \in \Cc$.}
		\For{$s \in \Cc$}
		    \State Initialize $K_s \leftarrow \varnothing$
		    \State Sort $i \in F'(s)$ in ascending order of $\min_{r \in N'(i)}D(r)$
			\For{$i \in F'(s)$}
			    \If{$y_i' = 0$} \State $\dy_i \leftarrow 0$
				\ElsIf{$|K_s| + 1 \leq \lfloor (1 + \eps)\sum_{u \in F'(s)}y_u'\rfloor$}
					\State $K_s \leftarrow K_s \cup \{i\}$, $\dy_i \leftarrow 1$
					\For{$j \in N'(i)$}
				        \State $\dx_{ij} \leftarrow x'_{ij}$
				    \EndFor
				\Else
					\State $\dy_i \leftarrow 0$
					\For{$r \in N'(i)$}
					    \For{$h \in K_s$}
					        \State $\dx_{ir} \leftarrow 0$
					        \State $\dx_{hr} \leftarrow \dx_{hr} + x'_{ir}/|K_s|$
					   \EndFor
					\EndFor
				\EndIf
			\EndFor
		\EndFor
		\Return $(\dx, \dy)$, $K_s$ for $s \in \Cc$.
	\end{algorithmic}
\end{algorithm}
\begin{lemma}\label{reroutlemma}
    After \cref{rerouting}, for every facility $h \in F'$, $\dy_h = 1$: $L(h, \dx) \leq 3(\nu + 4\rho\lambda)T$.
    Moreover, $\sum_{h \in F'}\dy_h \leq (1 + 3\eps)k$.
\end{lemma}
\begin{proof}
    Since for every $s \in \Cc$ we have $\sum_{u \in F'(s)}y_u' \geq 1/(1 + \eps)$, $\lfloor (1 + \eps)\sum_{u \in F'(s)}y_u'\rfloor \geq 1$ and $|K_s| \geq 1$.
    After filtering and preprocessing steps, $\sum_{u \in F'}y_u' \leq (1 + \eps)k$, so
    \[\sum_{h \in F'}\dy_h = \sum_{s \in \Cc}\sum_{h \in K_s}\dy_h \leq (1 + \eps)\sum_{s\in \Cc}\sum_{u \in F'(s)}y'_u = (1 + \eps)\sum_{u \in F'}y_u' \leq (1 + \eps)^2k \leq (1 + 3\eps)k. \]

    Next, let $h \in F'$, $\dy_h = 1$, and let $\Cc(h) = s$.
    Take $i \in F'(s)$ that was closed by \cref{rerouting}.
    The demand of every $r \in N'(i)$ served by $i$ gets split between all opened facilities from $K_s$ at an equal fraction.
    So, after we reroute $i$ into $h$, the \emph{additional} load of $h$ is
    \[\sum_{r \in N'(i)}d(h, r)\frac{x'_{ir}}{|K_s|} \leq \frac{1}{|K_s|}\sum_{r \in N'(i)}d(i, r)x'_{ir} + \frac{d(h, i)}{|K_s|}\sum_{r \in N'(i)}x'_{ir}.\]
    Recall that $\sum_{r \in N'(i)}d(i, r)x'_{ir} = L(i, x')\leq \nu Ty_i'$.
    Let $v = \argmin_{w \in N'(h)}D(w)$ and $j = \argmin_{r \in N'(i)}D(r)$.
    Since $h$ was opened, and $i$ was closed, $D(v) \leq D(j)$, and by \cref{dist} $d(h, i) \leq 6\rho D(j)$.
    Hence, the \emph{additional} load of $h$ is at most
    \begin{multline*}
        \frac{1}{|K_s|}\sum_{r \in N'(i)}d(i, r)x'_{ir} + \frac{d(h, i)}{|K_s|}\sum_{r \in N'(i)}x'_{ir} \leq \frac{\nu Ty_i'}{|K_s|} + \frac{6\rho D(j)}{|K_s|} \sum_{r \in N'(i)}x'_{ir}\leq \\
        \leq \frac{\nu Ty_i'}{|K_s|} + \frac{6\rho D(j)}{|K_s|} \sum_{r \in N'(i)}y_i'   =\frac{y_i'}{|K_s|}\left(\nu T + 6\rho \cdot |N'(i)|D(j) \right) \leq \\ \leq  \frac{y_i'}{|K_s|}\left(\nu T + 6\rho\cdot \lambda T \right) = \frac{y_i'}{|K_s|}(\nu + 6\rho\lambda)T.
    \end{multline*}
    We used the bound $|N'(i)|\min_{r \in N'(i)}D(r) \leq \sum_{r \in N'(i)}D(r) \leq \lambda T$ for non-$\lambda$-heavy facilities.
    Hence, the \emph{total additional} load of $h$, gained after rerouting all closed facilities $i \in F'(s)\setminus K_s$ in its cluster, is at most
    \begin{multline*}
        \sum_{i \in F'(s)\setminus K_s}\sum_{r \in N'(i)}d(h, r)\frac{x'_{ir}}{|K_s|} \leq \sum_{i \in F'(s)\setminus K_s}\frac{y_i'}{|K_s|}(\nu + 6\rho\lambda)T = \\ =(\nu + 6\rho\lambda)T\cdot \frac{\sum_{i \in F'(s)\setminus K_s}y_i'}{\lfloor (1 + \eps)\sum_{u \in F'(s)}y_u'\rfloor} \leq (\nu + 6\rho\lambda)T\cdot \frac{(1 + \eps)\sum_{i \in F'(s)}y_i'}{\lfloor (1 + \eps)\sum_{u \in F'(s)}y_u'\rfloor}\leq\\ \leq (2\nu + 12\rho\lambda)T.
    \end{multline*}
    The load of $h$ before rerouting was $L(h, x') \leq \nu Ty'_h \leq \nu T$, so after \cref{rerouting} the total load of facility $h$ is $L(h, \dx)\leq  3(\nu + 4\rho\lambda)T$.
    This holds for every $h \in K_s$ and every center $s \in \Cc$.
\end{proof}

Now we are ready to complete the analysis of the rounding algorithm.
\begin{proof}[Proof of \cref{scload}]
    We claim that, having completed all the intermediate steps from filtering and up to \cref{rerouting} included, with parameter values $\rho = \frac{1 + \eps}{\eps}$, $\gamma = \eps/(1 + \eps)$ and $\lambda = 1/\eps$, for the resulting solution $(\dx, \dy)$ it holds:
    \begin{enumerate}
        \item $\dy$ is integral, and $\sum_{i \in F}\dy_i \leq (1 + 4\eps)k$;
        \item for every $j \in C$, $1 \geq \sum_{i \in F}\dx_{ij} \geq 1/(1 + \eps)$, and if $j \in C\setminus C'$, $\sum_{i \in F}\dx_{ij} = 1$;
        \item for every $i \in F$, $L(i, \dx) \leq 12\left(1 +\frac{1 + \eps}{\eps^2}\right)T\dy_i$.
    \end{enumerate}  
	
	By \cref{heavylemma}, \cref{heavy} could open additional $\eps k$ facilities, so $\sum_{i \in F\setminus F'}\dy_i \leq \eps k$.
	By \cref{reroutlemma}, $\sum_{h \in F'}\dy_h \leq (1 + 3\eps)k$.
	This gives us total opening $\sum_{i \in F}\dy_i \leq (1 + 4\eps)k$.
	
	Next, take $j \in C$.
	If $j$ was serving some $\lambda$-heavy facility $i$, then $j \in C\setminus C'$, and \cref{heavy} sets $\dx_{ij} = 1$ and $\dx_{hj} = 0$ for all other facilities $h \neq i$.
	If $j$ did not serve any $\lambda$-heavy facility, then $j \in C'$, and we get $1 \geq \sum_{i \in F}\dx_{ij} = \sum_{i \in F}x'_{ij} \geq 1/(1 + \eps)$ after rerouting.
	
	Finally, if $i \in F$ was $\lambda$-heavy, by \cref{heavylemma} $L(i, \dx) \leq \frac{\nu}{\gamma}T \leq \frac{4}{\eps}T\dy_i$.
    Let $i$ be non-$\lambda$-heavy, I.e. $i \in F'$.
	If $\dy_i = 0$, I.e. $i$ is closed, then \cref{rerouting} assures that $L(i, \dx) = 0$.
	If $\dy_i = 1$, then by \cref{reroutlemma} we have  $L(i, \dx) \leq 3(\nu + 4\rho\lambda)T \leq 3\left(4 + 4\frac{1 + \eps}{\eps^2}\right)T\dy_i = 12\left(1 +\frac{1 + \eps}{\eps^2}\right)T\dy_i$.
	
	For every $(i, j) \in F' \times C'$, we multiply the assignment variables $\dx_{ij}$ by $1/\left(\sum_{i \in F}\dx_{ij}\right)$.
	Since $\sum_{i \in F}\dx_{ij} \geq 1/(1 + \eps)$, the load of every opened facility in $F'$ gets increased at most by a factor of $1 + \eps \leq 2$.
	After this change, $\sum_{i \in F}\dx_{ij} = 1$ and $L(i, \dx) \leq 24\left(1 + \frac{1 + \eps}{\eps^2}\right)T$ for all $j \in C$, $i \in F$.
	
	The solution $(\dx, \dy)$ has integral opening $\dy$, and every client $j \in C$ is served fully (I.e. $\sum_{i \in F}\dx_{ij} = 1$).
	By applying minimum makespan rounding algorithm \cite{LST90}, we get an integral assignment with respect to facilities opened in $\dy$, sacrificing another factor of 2 in approximation.
	We obtain a $\left(48\left(1 + \frac{1 + \eps}{\eps^2}\right), 1 + 4\eps\right)$-approximate solution to $\SC(T, k)$ problem, and the whole algorithm clearly runs in polynomial-time.
\end{proof}

\section{$(2 + \eps, O(1/\eps^2))$-approximation to $\SC(T, k)$}
\label{sec:algo2}

Similar to the $(O(1/\eps^2), 1 + \eps)$-approximation to $\SC(T, k)$, our goal is, given some solution $(x, y)$ to $\sclp(T, k)$, find an integral opening $\dy$ and fractional assignment $\dx$, and then apply minimum makespan rounding \cite{LST90}, which will prove \cref{scsize}.
However, this time we need to assure that $L(i, \dx) \leq (1 + \eps/2)T$ for every $i \in F$.
To achieve this, we use the same steps, applied in different order and with different values of parameters.

\subsection{Preprocessing and opening heavy facilities}

Let $(x, y)$ be a feasible fractional solution to $\sclp(T, k)$, and let $\eps \in (0, 1)$.
Straightahead, we apply preprocessing algorithm from \cref{preproc1} to $(x, y)$ with parameters $\mu = 1$ and $\gamma = \frac{1}{1 + \eps}$.
This will give us a solution $(x', y')$ such that 
\begin{enumerate}
    \item $y' = y$, and $\sum_{i \in F}y_i' \leq k$,
    \item for all $(i, j) \in F \times C$, $y'_i \geq x_{ij}'$ and if $x'_{ij} > 0$ then $x'_{ij} \geq \gamma y_i' = y_i'/(1 + \eps)$,
    \item for every $j \in C$, $1 \geq \sum_{i \in F}x_{ij}' \geq 1 - \gamma = \eps/(1 + \eps)$, and
    \item for every $i \in F$, $L(i, x') \leq (\mu + 2 - \gamma)Ty_i' = (2 + \frac{\eps}{1 + \eps})Ty_i' = \nu Ty_i'$.
\end{enumerate}

In this algorithm, we overuse the notation and define $D(j)$ with respect to assignment $x'$.
\begin{definition}
    For $j \in C$, let $D(j) := \sum_{i \in F}d(i, j)x'_{ij}$, the average facility distance to client $j$.
\end{definition}
The definitions of $N'(i)$, $N'(j)$ for $i \in F'$, $j \in C'$, given $F'\subseteq F$ and $C'\subseteq C$, are the same.
\begin{definition}
    For $F'\subseteq F$, $C'\subseteq C$, let $N'(i) :=  \{j \in C' : x'_{ij} > 0\}$, $N'(j) := \{i \in F' : x'_{ij} > 0\}$.
\end{definition}
\begin{definition}
    A facility $i \in F'$ is \emph{$\lambda$-heavy} for $\lambda > 0$, if $\sum_{j \in N'(i)}D(j) > \lambda T$.
\end{definition}
We apply \cref{heavy} to $(x', y')$ with $\lambda := \eps^2/15$.
Observe that
\[\sum_{j \in C}D(j) = \sum_{i \in F}\sum_{j \in C}d(i, j)x'_{ij} \leq \sum_{i \in F}\nu Ty_i' \leq \nu Tk.\]
Hence, applying a similar analysis as in \cref{heavylemma}, we open at most $\frac{\nu}{\lambda}k = O(k/\eps^2)$ additional facilities, and the load of every opened facility is at most $\frac{\nu}{\gamma}T = (1 + \eps)\left(2 + \frac{\eps}{1 + \eps}\right)T = (2 + 3\eps)T$.

For the returned sets $F'$ and $C'$, $\sum_{j \in N'(i)}D(j) \leq \lambda T = \eps^2T/15$, for every $i \in F'$.
As before, it remains to find the integral opening among facilities in $F'$.
However, there may be clients $j \in C'$, for which preprocessing step might have dropped a very huge portion of their demand, as the best bound we have is $\sum_{i \in F'}x'_{ij} \geq \eps/(1 + \eps)$.
Just for the same reason, the opening $\sum_{i \in N'(j)}y_i'$ may be too small for some clients $j \in C'$, so we cannot apply the clustering and rerouting steps to solution $(x', y')$, as we did in \cref{sec:algo1}, without loosing a lot in both approximation factors, we even do not have any distance upper bounds.
We are going to handle these issues by applying a specific filtering step to $(x', y')$, bounding the distance between facilities and clients they serve, as well retrieving the lost demand of every client in $C'$.

\subsection{Filtering}

We apply filtering to the restriction of $(x', y')$ on $F' \times C'$, however, the filtering process will be quite different from \cite{LV92}.
We will rely a lot on the fact that we now operate with non-$\lambda$-heavy facilities only.
\begin{definition}
    Let $\rho := \frac{(1 + \eps)^2}{\eps^2}$.
    For every $j \in C$ define $F'_j := \{i \in F' : d(i, j)\leq \rho D(j)\}$.
\end{definition}
\begin{lemma}\label{filter}
    For every $j \in C'$, $\sum_{i \in F'_j}x'_{ij} \geq 1/(\rho\eps) = \eps/(1 + \eps)^2$.
\end{lemma}
\begin{proof}
    Every $j \in C'$ was served by $F'$ only, therefore $D(j) = \sum_{i \in F'}d(i, j)x'_{ij}$.
    Observe that at most a portion of $1/\rho$ demand of $j$ can be served by facilities not in $F'_j$.
    Otherwise,
    \[D(j) = \sum_{i \in F'}d(i, j)x'_{ij} \geq \sum_{i \in F'\setminus F'_j}d(i, j)x'_{ij} \geq \rho D(j)\sum_{i \in F'\setminus F'_j}x'_{ij} > \rho D(j)\cdot \frac{1}{\rho} = D(j), \]
    a contradiction.
    Hence, $\sum_{i \in F'\setminus F'_j}x'_{ij} \leq 1/\rho$.
    Since $\sum_{i \in F'}x'_{ij} \geq \eps/(1 + \eps)$ for all $j \in C'$, we have
    \[\sum_{i \in F'_j}x'_{ij} = \sum_{i \in F'}x'_{ij} - \sum_{i \in F'\setminus F'_j}x'_{ij} \geq \frac{\eps}{1 + \eps} - \frac{\eps^2}{(1 + \eps)^2} = \frac{\eps}{(1 + \eps)^2} = \frac{1}{\rho\eps}.\]
\end{proof}
We construct a new solution $(\hx, \hy)$ as follows:
\[\text{for all $(i, j) \in F' \times C'$,}\qquad\qquad \hx_{ij} = \begin{cases} 
    0,& i \notin F'_j;\\
    \frac{x'_{ij}}{\sum_{i \in F'_j}x'_{ij}},& i \in F'_j;
\end{cases}\qquad \hy_i = \min\left(1, \rho\eps y_i'\right).\]
Clearly, $\sum_{i \in F'}\hy_i \leq \rho\eps \sum_{i \in F'}y_i' \leq \rho\eps k = O(k/\eps)$.
Also, by \cref{filter}, $\hx_{ij} \leq \min(1, \rho\eps x'_{ij}) \leq \hy_i$ for every $(i, j) \in F' \times C'$.
To bound $L(i, \hx)$ for $i \in F'$, recall that $i$ is non-$\lambda$-heavy, therefore $\sum_{j \in N'(i)}D(j) \leq \lambda T = \eps^2T/15$.
Since $\hx_{ij} > 0$ if and only if $x'_{ij} > 0$ and $d(i, j) \leq \rho D(j)$, 
\[\lambda T\hy_i \geq \sum_{j \in N'(i)}D(j)\hy_i \geq \sum_{\substack{j \in N'(i)\\s.t.\,\hx_{ij} > 0}}D(j)\hy_i \geq \frac{1}{\rho} \sum_{\substack{j \in N'(i)\\s.t.\,\hx_{ij} > 0}}d(i, j)\hy_i \geq \frac{1}{\rho} \sum_{\substack{j \in N'(i)\\s.t.\,\hx_{ij} > 0}}d(i, j)\hx_{ij},\]
implying
\[L(i, \hx) = \sum_{\substack{j \in N'(i)\\s.t.\,\hx_{ij} > 0}}d(i, j)\hx_{ij} \leq \rho\lambda T\hy_i= \frac{\rho\eps^2}{15}T\hy_i = \frac{(1 + \eps)^2}{15}T\hy_i.\]
Also, for every $j \in C'$ we now have $\sum_{i \in F'}\hx_{ij} = 1$ and 
$\sum_{i: \hx_{ij} > 0}\hy_i \geq 1$.

Since $\{i \in F' : \hx_{ij} > 0\}\subseteq N'(j)$ and $\{j \in C' : \hx_{ij} > 0\}\subseteq N'(i)$, we will abuse the notation and redefine $N'(i)$ and $N'(j)$ in terms of assignment $\hx$.
Let $\hat{\nu} := \frac{(1 + \eps)^2}{15}$.
It holds for $(\hx, \hy)$:
\begin{enumerate}
    \item $\sum_{i \in F'}\hy_i \leq \rho\eps k = O(k/\eps)$,
    \item for all $(i, j) \in F' \times C'$, $\hy_i \geq \hx_{ij}$ and if $\hx_{ij} > 0$ then $d(i, j) \leq \rho D(j)$,
    \item for every $j \in C'$, $\sum_{i \in F'}\hx_{ij} = 1$ and $\sum_{i \in N'(j)}\hy_i \geq 1$,
    \item for every $i \in F'$, $L(i, \hx) \leq \hat{\nu} T\hy_i$.
\end{enumerate}

\subsection{Finishing the algorithm}

Now we can correctly use our clustering and rerouting algorithms with $(\hx, \hy)$.
We subsequently apply \cref{clustering} and \cref{rerouting} to $(\hx, \hy)$ with newly defined sets $N'$ for $F'$ and $C'$, with corresponding values of of parameters $\lambda, \rho$ and $\nu \equiv \hat{\nu}$, obtaining the integral opening $\dy$ and possibly fractional assignment $\dx$ over $(F', C')$.
By \cref{reroutlemma}, for $h \in F'$: $\dy_h = 1$,
\[L(h, \dx) \leq 3(\hat{\nu} + 4\rho\lambda)T = 3\left(\frac{(1 + \eps)^2}{15} + 4\frac{(1 + \eps)^2}{\eps^2}\cdot \frac{\eps^2}{15}\right) = (1 + \eps)^2T \leq (1 + 3\eps)T,\]
and we open at most $(1 + \eps)\sum_{i \in F'}\hy_i = O(k/\eps)$ facilities.

Since for every $j \in C$ we have $\sum_{i \in F}\dx_{ij} = 1$, there is no need to modify fractional variables $\hx$ any further.
Observe that all $i \in F\setminus F'$ serve $j \in C\setminus C'$ only, these $j$ are assigned to $i \in F\setminus F'$ integrally, and for all $i \in F\setminus F'$ we have $L(i, \dx) \leq (2 + 3\eps)T$.
Therefore, it remains to obtain integral assignment over $(F', C')$, where for every $i \in F'$ we have $L(i, \dx) \leq (1 + 3\eps)T$.
By applying minimum makespan rounding algorithm \cite{LST90} to the restriction of $(\dx, \dy)$ on $(F', C')$, we get integral assignment, sacrificing a factor of 2 in load approximation for $i \in F'$, resulting in maximum load of the final solution at most $(2 + 6\eps)T$.
\cref{heavy} might have opened at most $O(k/\eps^2)$ additional facilities, so we obtain a $\left(2 + 6\eps, O(1/\eps^2)\right)$-approximate solution to $\SC(T, k)$ problem, and the whole algorithm clearly runs in polynomial-time, proving \cref{scsize}.

\section*{Acknowledgements}

We are very grateful to Ola Svensson for influential discussions at multiple stages of this work.

\bibliographystyle{alpha}
\bibliography{references}

\appendix
\newpage
\section{Preprocessing}
\label{sec:appa}
\begin{theor}[\cref{preproc1} restated]\label{preproc}
    Let $(x, y)$ be  such that, for all $i \in F$, $L(i, x) \leq \mu T y_i$ for some constant $\mu \geq 1$ and all other constraints of $\sclp(T, k)$ on variables $x$ are satisfied.
    There exists a polynomial time algorithm that, given such solution $(x, y)$ and a constant $\gamma \in (0, 1)$, finds a solution $(x', y')$ such that
    \begin{enumerate}
        \item $y' = y$, and if $x_{ij} = 0$, then $x'_{ij} = 0$;
        \item for every $(i, j) \in F\times C$, $y'_i \geq x'_{ij}$, and if $x'_{ij} > 0$, then $x'_{ij} \geq \gamma y_i'$;
        \item for every $j \in C$, $1 \geq \sum_{i \in F}x'_{ij} \geq 1 - \gamma$;
        \item for every $i \in F$, $L(i, x') \leq (\mu + 2 - \gamma)Ty_i'$.
    \end{enumerate}
\end{theor}
The algorithm we use in \cref{preproc1} is heavily inspired by the minimum makespan rounding algorithm, introduced by Lenstra et. al in \cite{LST90}.
In a sense, their algorithm achieves the desired property: in minimum makespan problem we have $y_i = 1$ for all $i \in F$, so for $j \in C$ we wish to have either $x_{ij}' = 0$ or $x_{ij}'= 1 = y_i'$.
The key difference is that in our case $y$ is \emph{not} integral, which requires several modifications of the original algorithm.

Let $(x, y)$ and $\gamma \in (0, 1)$ be given.
Let $\wF\subseteq F$, $\wC\subseteq C$, let $E \subseteq \wF \times \wC$.
Consider a bipartite graph $G = (\wF \cup \wC, E)$, let $\delta_E(v)$ be the neighbors of $v \in \wF \cup \wC$ in $G$, I.e., for $i \in \wF$, $\delta_E(i) = \{j \in \wC : (i, j) \in E\}$, and for $j \in \wC$, $\delta_E(j) = \{i \in \wF : (i, j) \in E\}$.
For $(i, j) \in \wF \times \wC$ we introduce a variable $w_{ij}$, and numbers $d_j \leq 1$ and $L_i \leq \mu Ty_i$, which can be thought of as the remaining demand of client $j \in \wC$ and the remaining load of facility $i \in \wF$ correspondingly.
Given sets $\wF, \wC, E$ and numbers $d, L$, we define the polytope $P(\wF, \wC, E, d, L)$ as the solution set of the following feasibility linear program:
\[\begin{aligned}
&&&\sum_{i \in \delta_E(j)}w_{ij} = d_j,&\forall j \in \wC,\\
&&&\sum_{j \in \delta_E(i)}d(i, j)w_{ij} \leq L_i,&\forall i \in \wF,\\
&&& w_{ij} \leq \min(y_i, d_j),&\forall (i, j) \in E,\\
&&& w_{ij} \geq 0,&\forall (i, j) \in E.
\end{aligned}\tag{$P(\wF, \wC, E, d, L)$}\]
Note that all values $y_i$ for $i \in F$ are fixed, so for every number $d_j$, $j \in C$, we have either constraint $\{w_{ij} \leq y_i\}$ or constraint $\{w_{ij} \leq d_j\}$.
The extreme points of $P(\wF, \wC, E, d, L)$ possess some very important properties, which resemble the properties of the extreme point solutions to the auxiliary program for the minimum makespan rounding algorithm of \cite{LST90}.
\begin{lemma}\label{extpoint}
    Let $w$ be an extreme point of $P(\wF, \wC, E, d, L)$, where $d_j \geq \gamma$ for all $j \in \wC$.
    One of the following must hold: 
    \begin{enumerate}[label=(\alph*)]
        \item there exists $(i, j) \in E$ such that $w_{ij} = 0$,
        \item there exists $(i, j) \in E$ such that $w_{ij} = y_i$,
        \item there exists $(i, j) \in E$ such that $w_{ij} = d_j$,
        \item there eixsts $i \in \wF$ such that $|\delta_E(i)| \leq 1$,
        \item there exists $i \in \wF$ such that $|\delta_E(i)| = 2$ and $\sum_{j \in \delta_E(i)}w_{ij} \geq \gamma y_i$.
    \end{enumerate}
\end{lemma}
\begin{proof}
    Suppose that none of (a), (b), (c), or (d) hold.
    We will show that (e) must hold then.
    
    For all $(i, j) \in E$ we have $0 < w_{ij} < \min(y_i, d_j)$, and for every $i \in \wF$ we have $|\delta_E(i)| \geq 2$.
    Since $\sum_{i \in \delta_E(j)}w_{ij} = d_j$ for all $j \in \wC$, we must also have $|\delta_E(j)| \geq 2$.
    As $w$ is an extreme point of $P(\wF, \wC, E, d, L)$, there exist $\wF_* \subseteq \wF$ and $\wC_* \subseteq \wC$ such that $\sum_{i \in \delta_E(j)}w_{ij} = d_j$ for all $j \in \wC_*$, $\sum_{j \in \delta_E(i)}d(i, j)w_{ij} = L_i$ for all $i \in \wF_*$, $|\wF_*| + |\wC_*| = |E|$, and constraints corresponding to $\wF_*, \wC_*$ are linearly independent.
    Since $2|E| = 2|\wF_*| + 2|\wC_*| \leq \sum_{i \in \wF_*}|\delta_E(i)| + \sum_{j \in \wC_*}|\delta_E(j)| \leq 2|E|$,
    for all $i \in \wF_*$ we must have $|\delta_E(i)| =2$, as well as $|\delta_E(j)| = 2$ for all $j \in \wC_*$.
    Therefore, the subgraph $G[\wF_*\cup \wC_*]$ of $G$ induced on $\wF_* \cup \wC_*$ is a bipartite union of disjoint cycles.
    
    Let $H$ be a cycle of $G[\wF_*\cup \wC_*]$, let $H_{\wF_*} := H\cap \wF_*$, $H_{\wC_*} := H \cap \wC_*$.
    Since for all $i \in H_{\wF_*}$ we have $|\delta_E(i)| = 2$, $\delta_E(i)\subseteq H \cap E$, and similarly, as $|\delta_E(j)| = 2$ for all $j \in H_{\wC_*}$, $\delta_E(j)\subseteq H \cap E$.
    Suppose that (e) does not hold, then for all $i \in H_{\wF_*}$ we have $\sum_{j \in \delta_E(i)}w_{ij} < \gamma y_i$.
    It follows that
    \[\sum_{i \in H_{\wF_*}}y_i> \frac{1}{\gamma}\sum_{i \in H_{\wF_*}}\sum_{j \in \delta_E(i)}w_{ij} = \frac{1}{\gamma}\sum_{(i, j) \in H\cap E}w_{ij} = \frac{1}{\gamma}\sum_{j \in H_{\wC_*}}\sum_{i \in \delta_E(j)}w_{ij} = \frac{1}{\gamma}\sum_{j \in H_{\wC_*}}d_j.\]
    The last inequality follows from $\sum_{i \in \delta_E(j)}w_{ij} = d_j$ for every $j \in \wC_*$.
    Since $d_j \geq \gamma$ for all $j \in \wC$, $d_j \geq \gamma y_i$ for all $(i, j)\in E$.
    Since $H$ is a cycle in bipartite graph, it has even length, its vertices alternate between $\wF_*$ and $\wC_*$, and $|H_{\wF_*}| = |H_{\wC_*}|$.
    Then, we can split the vertices of $H$ into disjoint consecutive pairs $(i, j)$, so that $i \in H_{\wF_*}$, $j \in H_{\wC_*}$, $(i, j) \in H\cap E$, and apply $d_j \geq \gamma y_i$ for every pair.
    Therefore, $\sum_{j \in H_{\wC_*}}d_j \geq \gamma \sum_{i \in H_{\wF_*}}y_i$, which combined with inequalities above leads to a contradiction.
    So, there must exist $i \in H_{\wF_*}$ such that $\sum_{j \in \delta_E(i)}w_{ij} \geq \gamma y_i$, implying (e).
\end{proof}

We transform $(x, y)$ into $(x', y')$ using a similar approach as in \cite{LST90}.
On every step $t \geq 1$ of the algorithm, we provide values of parameters $\wF^t, \wC^t, E^t, d^t, L^t$ so that polytope $P^t := P(\wF^t, \wC^t, E^t, d^t, L^t)$ is nonempty and $d^t_j \geq \gamma$ for $j \in \wC^t$, and find its extreme point $w^t$.
By \cref{extpoint}, either (a), (b), (c), (d) or (e) cases may occur for $w^t$.
If (a), we set $x'_{ij} \leftarrow 0$, $E^{t + 1} \leftarrow E^t\setminus \{(i, j)\}$.
If (b), we set $x'_{ij} \leftarrow y_i$, $d^{t + 1}_j \leftarrow d^t_j - y_i$, $L^{t + 1}_i \leftarrow L^t_i - d(i, j)w_{ij}^t$, $E^{t + 1} \leftarrow E^t\setminus \{(i, j)\}$.
If (c), we set $x'_{ij} \leftarrow d^t_j$, $d^{t + 1}_j \leftarrow 0$, $L^{t + 1}_i \leftarrow L^t_i - d(i, j)w^t_{ij}$, $E^{t + 1} \leftarrow E^t \setminus \{(i, j)\}$.
If (d) or (e), we set $\wF^{t + 1} \leftarrow \wF^t \setminus \{i\}$.
After processing exactly one case  (a), (b), (c), (d) or (e), we scan $j \in \wC^{t + 1}$, and if $d_j^{t + 1} < \gamma$ for some $j$, set $\wC^{t + 1} \leftarrow \wC^{t + 1} \setminus \{j\}$,  $x'_{ij} \leftarrow 0$ for all $(i, j)$ such that $i \in \delta_{E^{t + 1}}(j)$, and then $E^{t + 1}\leftarrow E^{t + 1} \setminus \{(i, j) : i \in \delta_{E^{t + 1}}(j)\}$.
If the change of $\wF^{t + 1}, \wC^{t + 1}, E^{t + 1}, d^{t + 1}$ or $L^{t + 1}$ is not mentioned for current case, the values are as in step $t$, so even though we drop facility $i$ from $\wF^t$ in case (d) or (e), the edges $(i, j)$ for $j \in \delta_{E^t}(i)$ are still kept in $E^{t + 1}$.
Having processed $\wC^{t + 1}$, if $E^{t + 1} \neq \varnothing$, we move to step $t + 1$ and consider $P^{t + 1}$.
\cref{preprocalgo} gives the full pseudocode, summarizing all the steps.
\begin{algorithm}[H]
		\caption{Preprocessing}\label{preprocalgo}
		\begin{algorithmic}[1]
			\Require{Initial values of $\wF$, $\wC$, $E$, $d$, $L$, parameter $\gamma \in (0, 1)$.}
			\Ensure{An assignment $x'$.}
			\While{$E \neq \varnothing$}
				\State Find an extreme point $w$ of $P(\wF, \wC, E, d, L)$
				\If{$\exists (i, j) \in E : w_{ij} = 0$} $x'_{ij} \leftarrow 0$, $E \leftarrow E\setminus \{(i, j)\}$
				\ElsIf{$\exists (i, j) \in E : w_{ij} = y_i$}
					\State $x'_{ij} \leftarrow y_i$, $d_j \leftarrow d_j - y_i$, $L_i \leftarrow L_i - d(i, j)w_{ij}$, $E \leftarrow E \setminus \{(i, j)\}$
				\ElsIf{$\exists (i, j) \in E : w_{ij} = d_j$}
					\State $x'_{ij} \leftarrow d_j$, $d_j \leftarrow 0$, $L_i \leftarrow L_i - d(i, j)w_{ij}$, $E\leftarrow E\setminus \left\{(i, j)\right\}$
				\ElsIf{$\exists i \in \wF : |\delta_E(i)| \leq 1$} $\wF \leftarrow \wF\setminus \{i\}$
				\ElsIf{$\exists i \in \wF : |\delta_E(i)| = 2$ \textbf{and} $\sum_{j \in \delta_E(i)}w_{ij} \geq \gamma y_i$} $\wF \leftarrow \wF \setminus \{i\}$
				\EndIf
				\For{$j \in \wC$}
					\If{$d_j < \gamma$}
						\State $\wC \leftarrow \wC \setminus \{j\}$, \textbf{for $i \in \delta_E(j)$ do} $x'_{ij} \leftarrow 0$, $E\leftarrow E\setminus \left\{(i, j)\right\}$
					\EndIf
				\EndFor
			\EndWhile
			\Return $x'$, extended to $F \times C$ by adding zero entries
		\end{algorithmic}
	\end{algorithm}

It is easy to see that if $P^t$ is nonempty and $d^t_j \geq \gamma$ for $j \in \wC^{t}$, the very same holds for $P^{t + 1}$ in the next step, unless $E^{t + 1} = \varnothing$.
Indeed, we manually assure that for all $j$ kept in $\wC^{t + 1}$ the condition $d^{t + 1}_j \geq \gamma$ must hold, and the restriction of $w^t$ to the set $E^{t + 1}\subseteq E^{t}$ is a feasible solution to $P^{t + 1}$, by construction of the algorithm.
Moreover, if we take $\wF^1 = F$, $\wC^1 = C$, $E^1 = \{(i, j) \in F\times C : x_{ij} > 0\}$, $d^1_j = 1$ for $j \in C$ and $L^1_i = \mu T y_i$ for $i \in F$, $d^1_j \geq \gamma$ and $P^1$ is nonempty, since there is a feasible solution $w_{ij} := x_{ij}$ for $(i, j) \in E^1$.
We run \cref{preprocalgo} with these initial values of $\wF$, $\wC$, $E$, $d$ and $L$ given as input, obtaining an assignment $x'$.
By setting $y' := y$, we obtain a solution $(x', y')$.

We claim that \cref{preprocalgo} runs in polynomial-time, and solution $(x', y')$ satisfies all requirements of \cref{preproc1}.
By \cref{extpoint}, on every step $t \geq 1$ either (a), (b), (c), (d) or (e) must occur for $w^t$, the extreme point of $P^t$.
Then, either $|E^t|$, $|\wF^t|$ or $|\wC^t|$ is reduced at least by 1 after step $t$.
So, since $|E^1| \leq |F||C|$, after at most $2|F||C|$ steps we will have $E^{t + 1} = \varnothing$ for some $1 \leq t \leq 2|F||C|$.
Each step $t$ takes only polynomial time to perform, thus the total running time is also polynomial.

Since $E^1 = \{(i, j) \in F\times C : x_{ij} > 0\}$, the only positive coordinates of $x'$ can be $(i, j)$ such that $x_{ij} > 0$, as if $x_{ij} = 0 \iff (i, j) \notin E^1$, \cref{preprocalgo} sets $x'_{ij} = 0$ in the very end.
The constraint $\{w_{ij} \leq \min(y_i, d_j)\}$ of $P(\wF, \wC, E, d, L)$ assures that $x'_{ij} \leq y_i'$, for all $(i, j) \in F\times C$.
If $x'_{ij} > 0$, then either $x'_{ij} = y_i'$ (case (b)) or $x'_{ij} = d^t_j$ for some step $t \geq 1$ (case (c)).
Since $y_i' \leq 1$ and for all steps $t \geq 1$ we maintain $d_j^t \geq \gamma$ for all $j \in \wC$, in both cases we have $x'_{ij} \geq \gamma y_i'$.

Next, if after processing cases for $w^t$ during some step $t \geq 1$ we end up with $d^{t + 1} < \gamma$, client $j$ gets discarded from $\wC^{t + 1}$.
Since $d_j^1 = 1$ initially, by the end of step $t$ we must have assigned at least $1 - \gamma$ portion of $j$'s demand before discarding $j$ to make $d^{t + 1}_j < \gamma$.
Then, after \cref{preprocalgo} finishes, for all $j \in C$ we have $1 \geq \sum_{i \in F}x'_{ij} \geq 1 - \gamma$.

Finally, fix $i \in F$.
Observe that if $i \in \wF^t$ in the beginning of step $t \geq 1$, then 
\[\sum_{j \in \delta_{E^t(i)}}d(i, j)w^t_{ij}  \leq L_i^t = L^1_i -  \sum_{C\setminus \wC^t}d(i, j)x'_{ij}  \implies \sum_{C\setminus \wC^t}d(i, j)x'_{ij} + \sum_{j \in \delta_{E^t(i)}}d(i, j)w^t_{ij} \leq \mu T y_i,\]
by feasibility of $w^t$ for polytope $P^t$.
Suppose that after step $t$ facility $i$ gets removed from $\wF^t$, so $i\notin \wF^{t + 1}$.
If case (d) occurred and $|\delta_{E^t}(i)| \leq 1$, let $j \in \delta_{E^t}(i)$ be a single client served by facility $i$.
After removing $i$ from $\wF^t$, the constraint $\{\sum_{j \in \delta_E(i)}d(i, j)w_{ij} \leq L_i\}$ is not present in $P^{t + 1}$ and all future-step polytopes.
So, for any step $r \geq t + 1$, the load we may get after obtaining $w^{r}$ and determining the value of $x'_{ij}$ is at most $d(i, j)w^{r} \leq d(i, j)y_i \leq Ty_i$ (as $d(i, j) \leq T$ for all $x_{ij} > 0$).
The total load of facility $i$ becomes $L(i, x') \leq (\mu + 1)Ty_i$.

If case (e) occurred for this facility $i$, $|\delta_{E^t}(i)| = 2$ and $\sum_{j \in \delta_{E^t}(i)}w^t_{ij} \geq \gamma y_i$.
Let $j'$ and $j''$ be the two clients belonging to $\delta_{E^t}(i)$.
Their contribution to facility $i$'s load on step $t$ is exactly $d(i, j')w_{ij'}^t + d(i, j'')w_{ij''}^t$, which is at most $L_i^t$.
After removing $i$ from $\wF^t$, the constraint $\{\sum_{j \in \delta_E(i)}d(i, j)w_{ij} \leq L_i\}$ is not present in $P^{t + 1}$ and all future-step polytopes.
So, for any step $r \geq t + 1$, the load we may get after obtaining $w^{r}$ and determining the values of both $x'_{ij'}$ and $x'_{ij''}$ is at most $d(i, j')w_{ij'}^r + d(i, j'')w_{ij''}^r \leq d(i, j')y_i + d(i, j'')y_i$.
Hence, the \emph{additional} load facility $i$ gained since the end of step $t$ is at most 
\begin{multline*}
    (d(i, j')y_i + d(i, j'')y_i) -  (d(i, j')w_{ij'}^t + d(i, j'')w_{ij''}^t) =\\ = d(i, j')(y_i - w_{ij'}^t) + d(i, j'')(y_i - w_{ij''}^t) \leq T(2y_i - (w_{ij'}^t + w_{ij''}^t))\leq \\ \leq T(2y_i - \gamma y_i) = (2 - \gamma)Ty_i.
\end{multline*}
Therefore, the total load of facility $i$ becomes $L(i, x') \leq (\mu + 2 - \gamma)Ty_i$.

As a result, solution $(x', y')$ and the preprocessing algorithm (\cref{preprocalgo}) indeed satisfy all the claimed properties of \cref{preproc1}, thus finishing the proof.

\section{Hard instances}
\label{sec:appb}

We first present a hard instance for $\MLkSC$ problem.
Let $R, M$ be integers, $R \ll M$.
Let $k = 2R - 1$, $|F| = 2R$, $|C| = (M + R)R$.
$F$ and $C$ are partitioned into $R$ disjoint groups, each has exactly $2$ facilities and exactly $M + R$ clients.
For $i, h \in  F$, $d(i, h) = 1$ if $i, h$ are in the same group, otherwise $d(i, h) = R$.
In every group, one facility has $M$ collocated clients (call it $M$-facility), the other has $R$ collocated clients ($R$-facility).
The instance is illustrated in \cref{fig:igap1}.
\begin{figure}[H]
	\begin{center}
		\begin{tikzpicture}
		
		\draw[thick] (0,0) ellipse (2cm and 3cm);
		\node[above] at (-2.25,2.25) {\text{One Group}};
		\node[right] at (2,0){$\times R$};
		
		\draw[thick] (0, 1.5) ellipse (.85cm and .85cm);
		\draw[thick] (0, -1.5) ellipse (0.65cm and 0.65cm);
		
		\draw (-0.35,2) node[shape=circle,draw,fill=red,minimum size=6pt,inner sep=0pt]  {};
		\draw (0.35,2) node[shape=circle,draw,fill=red,minimum size=6pt,inner sep=0pt]  {};
		\draw (-0.35,1) node[shape=circle,draw,fill=red,minimum size=6pt,inner sep=0pt]  {};
		\draw (0.35,1) node[shape=circle,draw,fill=red,minimum size=6pt,inner sep=0pt]  {};
		
		
		\draw (0,1.5) node[draw,fill=green,minimum size=12pt,inner sep=0pt] (M) {};
		\node[above] at (0, 2.35) {$M$};
		
		\draw (0,-1.5) node[draw,fill=green,minimum size=8pt,inner sep=0pt] (R) {};
		\node[below] at (0, -2.15) {$R$};
		
		\draw (-0.25,-1.9) node[shape=circle,draw,fill=red,minimum size=4pt,inner sep=0pt]  {};
		\draw (0.25,-1.9) node[shape=circle,draw,fill=red,minimum size=4pt,inner sep=0pt]  {};
		\draw (-0.25,-1.1) node[shape=circle,draw,fill=red,minimum size=4pt,inner sep=0pt]  {};
		\draw (0.25,-1.1) node[shape=circle,draw,fill=red,minimum size=4pt,inner sep=0pt]  {};
		
		\draw[very thick] (0, -0.85) -- (0, 0.65);
		\node[right] at (0,0) {$1$};
		
		\end{tikzpicture}
	\end{center}
	\caption{Hard instance for $\mlklp$.}\label{fig:igap1}
\end{figure}
	
There is a feasible fractional solution to $\mlklp$ for this instance with $T = 1$.
Open every $M$-facility fully, and there assign all its collocated clients.
Next, open every $R$-facility to $1 - 1/R$, and let it serve $(1 - 1/R)$-fraction of its collocated clients' demand.
The remaining $1/R$ fraction of these clients' demand will be served by $M$-facility of the same group.
It is easy to see that the load of every $R$-facility is $0$, the load of every $M$-facility is $R\cdot 1/R \cdot 1 = 1$, and the opening is exactly $R\cdot (1 + 1 - 1/R) = 2R - 1 = k$.

Consider any integral solution to this instance of $\MLkSC$.
If it assigns some client to a facility from different group, maximal load will be at least $R$.
Suppose that all clients are assigned to facilities only from the same group.
Since $k = 2R - 1$, there will be at least one group with at most one facility opened, take this group.
If $M$-facility is opened, both its clients and clients of $R$-facility must be assigned to $M$-facility fully, resulting in its load $R \cdot 1 \cdot 1 = R$.
Similarly, if $R$-facility is opened, maximum load will be at least $M\gg R$.
Hence, the load of any integral star cover of size $k$ is at least $R$.
Furthermore, even if we allow opening $(1 + \eps)k$ facilities for $\eps = 1/(2R)$, since
\[(1 + \eps)k = \left(1 + \frac{1}{2R}\right)(2R - 1) = 2R - \frac{1}{2R} < 2R,\]
there will still be a group with at most one facility opened, resulting in maximum load at least $R = T/(2\eps)$, where $T = 1$ is maximal fractional load.
It follows that if $T^*$ is an optimal load to $\mlklp$, any integral $(1 + \eps)k$ star cover of $(F, C)$ has load is at least $\Omega(1/\eps)T^*$.
	
Now, we move to a hard instance for $\MSSC$.
For integer $N$, let $|F| = N$ and $|C| = N + 1$, the load bound $T \geq 1$ is arbitrary.
Both $F = \{i_1, \ldots, i_N\}$ and $C = \{J, j_1,\ldots, j_N\}$ are vertices of a bipartite graph, and the metric $d$ is a shortest-path metric.
For every $1 \leq r \leq N$ we have an edge $(i_r, j_r)$ or length $d(i_r, j_r) = (1 - 1/N)T$.
Also, every facility $i_r$ for $1 \leq r \leq N$ is connected to a ``central'' client $J$ by an edge of length $d(i_r, J) = T$.
The instance is illustrated in \cref{fig:igap2}.
\begin{figure}[H]
	\begin{center}
	    \begin{tikzpicture}
	        \draw (0,0) node[draw,fill=green,minimum size=12pt,inner sep=0pt] (i1)[label=left:$i_1$] {};
	        \draw (0,2.5) node[shape=circle,draw,fill=red,minimum size=12pt,inner sep=0pt] (j1) [label=left:$j_1$] {};
	        
	        \draw (2,0) node[draw,fill=green,minimum size=12pt,inner sep=0pt] (i2) [label=left:$i_2$] {};
	        \draw (2,2.5) node[shape=circle,draw,fill=red,minimum size=12pt,inner sep=0pt] (j2) [label=left:$j_2$] {};
	        
	        \draw (6,0) node[draw,fill=green,minimum size=12pt,inner sep=0pt] (i3) [label=right:$i_{N - 1}$] {};
	        \draw (6,2.5) node[shape=circle,draw,fill=red,minimum size=12pt,inner sep=0pt] (j3) [label=right:$j_{N - 1}$] {};
	        
	        \draw (8,0) node[draw,fill=green,minimum size=12pt,inner sep=0pt] (i4) [label=right:$i_{N}$] {};
	        \draw (8,2.5) node[shape=circle,draw,fill=red,minimum size=12pt,inner sep=0pt] (j4) [label=right:$j_{N}$] {};
	        
	        \draw[very thick] (i1) -- (j1);
	        \draw[very thick] (i2) -- (j2);
	        \draw[very thick] (i3) -- (j3);
	        \draw[very thick] (i4) -- (j4);
	        
	        \node[left] at (0, 1.25) {$(1 - 1/N)T$};
	        \node[right] at (8, 1.25) {$(1 - 1/N)T$};
	         
	        \draw (4,-2) node[shape=circle,draw,fill=red,minimum size=12pt,inner sep=0pt] (J) [label=below:$J$] {};
	        
	        \draw[very thick] (i1) -- (J);
	        \draw[very thick] (i2) -- (J);
	        \draw[very thick] (i3) -- (J);
	        \draw[very thick] (i4) -- (J);
	        
	        \draw[very thick] (J) -- (2.75, -0.25);
	        \draw[very thick] (J) -- (3.5, -0.25);
	        \draw[very thick] (J) -- (4, -0.25);
	        \draw[very thick] (J) -- (4.5, -0.25);
	        \draw[very thick] (J) -- (5.25, -0.25);
	        
	        \node at (3, 2.5) {$\textbf{\dots}$};
	        \node at (4, 2.5) {$\textbf{\dots}$};
	        \node at (5, 2.5) {$\textbf{\dots}$};
	        
	        \node at (3, 0) {$\textbf{\dots}$};
	        \node at (4, 0) {$\textbf{\dots}$};
	        \node at (5, 0) {$\textbf{\dots}$};
	        
	        \node[left] at (2, -1.35) {$T$};
	        \node[right] at (6, -1.35) {$T$};
	    \end{tikzpicture}
	\end{center}
	\caption{Hard instance for $\mslp$.}
	\label{fig:igap2}
\end{figure}

It is easy to see that in any integral solution to $\mslp$ every client $j_r$ for $1 \leq r \leq N$ can be served only by facility $i_r$.
Furthermore, client $J$ should also be served fully, so it should be assigned to one of $i \in F$.
Therefore, even if we open all facilities in $F$ fully, for some facility $i \in F$ which gets $J$ assigned to it, the load will be at least $(2 - 1/N)T$.
This means that there is no feasible integral solution to $\mslp$, and any integral solution violates the maximum load constraint at least by a factor of $(2 - 1/N)$.

On the other hand, there exists a feasible \emph{fractional} solution to $\mlklp$ for this instance.
We open all $i_r$ for $1 \leq r \leq N$ and assign $j_r$ fully to it.
Also, client $J$ gets served by all $i \in F$ at an equal fraction of $1/N$.
In this solution, the load of every facility $i \in F$ is exactly $T$.

\end{document}